\documentclass[preprint,11pt]{elsarticle}

\usepackage{amssymb}
\usepackage{amsthm}

\usepackage{framed}
\usepackage{fullpage}

\usepackage{amsmath}
\usepackage{tikz}
\tikzstyle{every node}=[circle,draw=black, inner sep=1.5pt,fill=white]
\usepackage{subfigure}

\newcommand{\mincov}{{\sf mincov}}
\newcommand{\maxcov}{{\sf maxcov}}
\newcommand{\cov}{{\sf cov}}

\newtheorem{theorem}{Theorem}
\newtheorem{lemma}{Lemma}
\newtheorem{corollary}[theorem]{Corollary}
\newtheorem{problem}{Problem}

\newcommand{\wt}{\;:\;}
\newcommand{\st}{\;|\;}
\renewcommand{\leq}{\leqslant}

\renewcommand{\geq}{\geqslant}

\newcommand{\f}
{\ensuremath{\varphi}}

\begin{document}

\begin{frontmatter}

%% Title, authors and addresses

%% use the tnoteref command within \title for footnotes;
%% use the tnotetext command for theassociated footnote;
%% use the fnref command within \author or \address for footnotes;
%% use the fntext command for theassociated footnote;
%% use the corref command within \author for corresponding author footnotes;
%% use the cortext command for theassociated footnote;
%% use the ead command for the email address,
%% and the form \ead[url] for the home page:
%% \title{Title\tnoteref{label1}}
%% \tnotetext[label1]{}
%% \author{Name\corref{cor1}\fnref{label2}}
%% \ead{email address}
%% \ead[url]{home page}
%% \fntext[label2]{}
%% \cortext[cor1]{}
%% \address{Address\fnref{label3}}
%% \fntext[label3]{}

\title{Interval scheduling maximizing minimum coverage}

%% use optional labels to link authors explicitly to addresses:
%% \author[label1,label2]{}
%% \address[label1]{}
%% \address[label2]{}

\author{Veli M\"akinen}
\author{Valeria Staneva\fnref{label1}}
\author{Alexandru I. Tomescu}
\author{Daniel Valenzuela\corref{cor1}}

\address{Helsinki Institute for Information Technology HIIT \\ Department of Computer Science \\ University of Helsinki, Finland}

\fntext[label1]{Current affiliation MIT -- Massachusetts Institute of Technology (student). Work conducted while visiting University of Helsinki as summer intern.} 
\cortext[cor1]{Corresponding author \ead{dvalenzu@cs.helsinki.fi}}

\begin{abstract}
In the classical interval scheduling type of problems, a set of $n$ jobs, characterized by their start and end time, need to be executed by a set of machines, under various constraints. In this paper we study a new variant in which the jobs need to be assigned to at most $k$ identical machines, such that the minimum number of machines that are busy at the same time is maximized. 
This is relevant in the context of genome sequencing and haplotyping, specifically when a set of DNA reads aligned to a genome needs to be pruned so that no more than $k$ reads overlap, while maintaining as much read coverage as possible across the entire genome. 
We show that the problem can be solved in time $\min\left(O(n^2\log k / \log n),O(nk\log k)\right)$ by using max-flows. We also give an $O(n\log n)$-time approximation algorithm with approximation ratio $\rho =\frac{k}{\lfloor k/2 \rfloor}$. 
%Interval scheduling is a popular problem where each job has
%a specific starting and ending time when it needs to be performed.
%Many  different variants can be find in the litarature.
%In this paper we address a yet unstudied problem variant
%of assigning a set of jobs to $k$ identical 
%machines, such that the maximum numbers of machines that 
%are idle at the same time is minimized.
%This problem is of actual relevance in the context of genome
%sequencing and haplotyping, specifically when a set of reads 
%needs to be pruned so that no more than $k$ reads overlap.

%% Text of abstract

\end{abstract}

\begin{keyword}
Interval scheduling \sep read pruning \sep haplotype assembly \sep max-flows
%% keywords here, in the form: keyword \sep keyword

%% PACS codes here, in the form: \PACS code \sep code

%% MSC codes here, in the form: \MSC code \sep code
%% or \MSC[2008] code \sep code (2000 is the default)

\end{keyword}

\end{frontmatter}

%% \linenumbers

%% main text
\section{Introduction}
% !TEX root = paper.tex
Interval scheduling is a classical problem in combinatorial optimization. The input usually consists of $n$ jobs, such that each job $j$ needs to be executed in the time interval $[s_j , f_j)$, by any available machine. In the most basic variant of this problem, each machine is always available, can process at most one job at a time, and once it starts executing a job it does so until it is finished. The task is to process all jobs using the minimum number of machines~\cite{NAV07}. This is solvable in time $O(n\log n)$~\cite{1675260}. In another problem variant, known as \emph{interval scheduling with given machines},
%known as the ``the $k$-track assignment problem'' or ``Interval scheduling on identical machines'', 
there are only $k$ available machines, and the execution of each job brings a specified profit. The task is to schedule a maximum-profit subset of jobs. This is also solvable in polynomial time, for example by min-cost flows~\cite{arkin87,bouzina96}. Some problem variants are NP-hard, for example if each job can be executed only by a given subset of machines~\cite{arkin87}, or if each machine is available during a specific period of time~\cite{brucker94}. See the surveys~\cite{NAV07,kovalyov07} for further references.
% other e.g.:
%``fixed interval scheduling'' problem~\cite{kovalyov07} where each job has a list
% of possible time intervals where it can be executed.
% * Machines with paralell capacity that can process up to g jobs each, 

Most previous work has focused on either maximizing the profit obtained from executing the jobs, or on minimizing the resources used by the jobs. In this paper we study a new problem variant with a rather different objective function, motivated by a new application of interval scheduling in genome haplotyping with high-throughput DNA sequencing. In this variant, which we call \emph{interval scheduling maximizing minimum coverage}, 
we need to select a subset of jobs to be executed by a given number $k$ of machines, such that the minimum, over the number of machines that are busy at any given time, is as large as possible. Fig.~\ref{fig:example} gives an example. To the best of our knowledge this variant has not been addressed before. 

\begin{figure}[t]
\centering
\includegraphics[width=0.9\columnwidth]{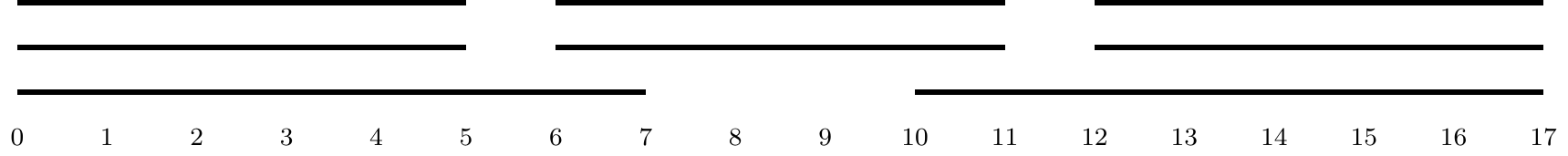}
\caption{An instance of interval scheduling maximizing minimum coverage in which 8 intervals are given and  $k = 2$ machines are available to execute them. Observe that in all solutions to this problem three disjoint intervals of length 5 need to be removed, leading to a solution that executes $5$ jobs and the number of idle machines is never greater than $1$. However, in the solutions to the classical interval scheduling with given machines problem, 
the two intervals of length 7 are removed both in the case when all intervals have the same profit, and in the case when the profit of an interval equals its length.\label{fig:example}}
\end{figure}

%% This could be here, or in the Pruning section  when the problem is formally introduced,
%  but either way , it needs to fit in the story. (or we can take it out :) )

%Some variants that optimize a different function are studied recently\cite{Mer15},
%where each machine can process more than one job, and the amount of busy time 
%is minimized.
% We are optimizing the opposite :)

The rest of the paper is structured as follows. In Sec.~\ref{sec:pruning} we discuss our original motivation in the context of high-throughput DNA sequencing and give the precise problem formulation. In Sec.~\ref{sec:max-flow} we present a reduction to a max-flow problem, which leads to a $O(n^2 \log k/\log n)$ solution for our problem. In Sec.~\ref{sec:tailored-flow} we present a tailored max-flow algorithm that runs in $O(n k \log k)$ time, which is faster than the previous when $k=o(n/\log n)$. Since for large $k$ the best complexity is almost quadratic, we also study a way to find approximate solutions: In Sec.~\ref{sec:2-approximation} we present an $O(n \log n)$-time  $\frac{k}{\lfloor k/2 \rfloor}$-approximation algorithm. 
%We implemented that version and released the code in ANONYMOUS WEBSITE.

\section{Haplotype phasing, read pruning and interval scheduling\label{sec:pruning}}

High-throughput sequencing is a technique developed over the last decade that can produce millions of DNA fragments, called \emph{reads}, from random positions across the genome of an individual. Depending on the technology, their length can be from hundreds to thousands of  characters. Many analyses are carried out by first aligning the reads to a reference genome sequence of the species, and studying, for example, the genetic variations of the individual with respect to the reference (see e.g.~\cite{makinen2015genome}). A more detailed analysis, called \emph{haplotype phasing}, also takes into account the fact that in some species, such as humans, each chromosome is present in two copies, inherited from each parent. In this context it is also desirable to assign the genetic variations to the copy of the chromosome where they are present.

Since real data has sequencing and alignment errors, a well-known problem formulation asks for the minimum number of corrections that enables a consistent partitioning of the input set of reads into the two copies of the chromosome they were sequenced from. This problem is called \emph{minimum error correction} and was introduced by Lippert \emph{et al.} in~\cite{lippert02} and proved NP-hard in~\cite{cilibrasi05}. A practical algorithm for this problem was proposed in~\cite{Pat14}, having a time complexity of $O(2^{k-1} m)$, where $m$ is the proportional to the length of the genome, and $k$ is the maximum number of reads covering any position of the genome. This algorithm is particularly useful because its runtime is independent of the read length. 

The higher the number of reads, and the more uniform they are distributed across the genome, the more accurate the solution to the minimum error correction problem is in practice. However, the $O(2^{k-1} m)$ time complexity makes this algorithm feasible only for small values of $k$. In its implementation~\cite{Pat14}, for every genomic position with too high read coverage, some reads are removed at random. However, this may arbitrarily lead to some other positions having a too low coverage for accurate results. In this paper we study the problem of pruning the read set such that the maximum read coverage is less than a given integer $k$, and the minimum coverage across all genomic positions is as high as possible. 

Our formal definition is as follows. We will represent each read $i$ as an interval~$[s_i, f_i)$. We will assume that $0 \leq s_i < f_i < N$.
%Note that in the genomic context, all numbers are integers and $N$ stands for the length of the reference genome,
%and numbers in [0,N) represents positions in the genome. 
%However we don't need to keep this assumptions to solve the scheduling problem. 
Given an interval $[s_i,f_i)$ and a point $p \in [s_i,f_i)$, we say that $[s_i,f_i)$ \emph{covers} $p$. If $p \in (s_i,f_i)$ we say that $[s_i,f_i)$ \emph{strictly covers} $p$. 
Given a set $S = \{ [s_i,f_i) \wt i \in \{1,\dots,n\} \st s_i < f_i\}$ of intervals, and a point $p$ 
we define the \emph{coverage of $p$} as $\cov_S(p) = |\{ [s_i,f_i) \in S |  [s_i,f_i) \text{ covers }  p  \} |$. 
When clear from the context, the subscript $S$ will be omitted.
We also define the  \emph{maximum coverage of $S$} as $\maxcov(S) = max_{p \in [0,N)}{ \cov_S(p)  }$.
Likewise, we define the \emph{minimum coverage of $S$} as $\mincov(S) = min_{p \in [0,N)}{ \cov_S(p)  }$. Our problem is the following one.

\begin{framed}
\begin{problem}[Interval scheduling maximizing minimum coverage]\ \newline
\label{prob:main-problem}
\noindent{\rm\!\textbf{INPUT.}} A set $S = \{ [s_i,f_i) \wt i \in \{1,\dots,n\} \st s_i < f_i\}$ of intervals and an integer $k$.\\
\noindent{\rm\textbf{TASK.}} Find an $S' \subseteq S$ such that $\maxcov(S') \leq k$ and maximizing $\mincov(S')$.
\end{problem}
\end{framed}

Note that if we only keep the first condition, namely $S' \subseteq S$ such that $\maxcov(S') \leq k$ 
our problem would be exactly the one of finding a feasible set of jobs to be scheduled on $k$ machines.

%TODO: better problem name!
\section{The reduction to max-flows\label{sec:max-flow}}
% !TEX root = paper.tex

In this section we show that the problem is solvable in time $O(n^2\log k/\log n)$ by max-flows. First, we consider the decision version of the maximization problem, as follows.

\begin{framed}
\begin{problem}[Interval scheduling with bounded coverage]\ \newline
\label{prob:decision-version}
\noindent{\rm\textbf{\!INPUT.}} A set $S = \{ [s_i,f_i) \wt i \in \{1,\dots,n\} \st s_i < f_i\}$ of intervals, and integers $k,t$.\\
\noindent{\rm\textbf{TASK.}} Decide if there an $S' \subseteq S$ such that $\maxcov(S') \leq k$ and $\mincov(S') \geq t$, and if yes, output such $S'$.
\end{problem}
\end{framed}

%Observe that if Problem~\ref{prob:decision-version} admits a solution for $t = t'$, then it admits a solution for all $t \in \{0,\dots,t'\}$. Therefore, Problem~\ref{prob:main-problem} can be solved by $O(\log k)$ instances of Problem~\ref{prob:decision-version}, by doing a binary  search for $k' \in \{1,\dots,k\}$, and for each such $k'$ solving Problem~\ref{prob:decision-version} with $t = k'$. 
%
We now describe the reduction of Problem~\ref{prob:decision-version} to a max-flow problem, following standard network flow notation, as can be found e.g.~in~\cite{makinen2015genome}. See also Fig.~\ref{fig:flow-b} for an example. Let $s = \min_{i \in \{1,\dots,n\}} s_i$ and $f = \max_{i \in \{1,\dots,n\}} f_i$. We construct a graph $G_{S,k,t}$ (possibly having parallel arcs) whose vertex set equals $\{s-1,f+1\} \cup \{s_i,f_i \wt i \in \{1,\dots,n\}\}$. Vertex $s-1$ will be the unique source of the graph, and vertex $f+1$ will be the unique sink of the graph. For every two consecutive numbers $i,j$ in $V(G)$ (i.e., such that there is no number $p$ in $V(G)$ with $i < p < j$), we add the arc $(i,j)$ to $E(G)$. We call these arcs \emph{backbone} arcs. The backbone arcs $(s-1,s)$ and $(f,f+1)$ receive capacity $k$, and the other backbone arcs have capacity $k - t$. For every interval $[s_i,f_i) \in S$, we add to $E(G)$ the arc $(s_i,f_i)$ with capacity $1$. We call these latter arcs \emph{interval} arcs. 

\begin{figure}[t]
\centering
\subfigure[\label{fig:flow-a}]{
\begin{tikzpicture}[scale=1,-]
\draw[clip,draw=none] (-1.35,-1.1) rectangle (11.3,1.1);
%%%%% backbone
\node[draw=none] (0) at (0,0) {\footnotesize $0$};
\node[draw=none] (1) at (1,0) {\footnotesize $1$};
\node[draw=none] (2) at (2,0) {\footnotesize $2$};
\node[draw=none] (3) at (3,0) {\footnotesize $3$};
\node[draw=none] (4) at (4,0) {\footnotesize $4$};
\node[draw=none] (5) at (5,0) {\footnotesize $5$};
\node[draw=none] (6) at (6,0) {\footnotesize $6$};
\node[draw=none] (7) at (7,0) {\footnotesize $7$};
\node[draw=none] (8) at (8,0) {\footnotesize $8$};
\node[draw=none] (9) at (9,0) {\footnotesize $9$};
\node[draw=none] (10) at (10,0) {\footnotesize $10$};
\draw[ultra thick] (0,1) to (8,1);
\draw[ultra thick] (0,-0.5) -- (1.90,-0.5);
\draw[ultra thick] (2,-0.5) -- (6,-0.5);
\draw[ultra thick] (1,-1) -- (10,-1);
\draw[ultra thick] (1,0.5) -- (3,0.5);
\draw[ultra thick] (4,0.5) -- (10,0.5);
\end{tikzpicture}}
%\\\vspace{0.5cm}
\subfigure[\label{fig:flow-b}]{
\begin{tikzpicture}[scale=1,->]
\draw[clip,draw=none] (-1.35,-1.3) rectangle (11.3,1.8);
%%%%% backbone
\node[circle,draw=black,inner sep=-2pt,fill=white,minimum size=5mm] (-1) at (-1,0) {\footnotesize $-1$};
\node[circle,draw=black,inner sep=-2pt,fill=white,minimum size=5mm] (0) at (0,0) {\footnotesize $0$};
\node[circle,draw=black,inner sep=-2pt,fill=white,minimum size=5mm] (1) at (1,0) {\footnotesize $1$};
\node[circle,draw=black,inner sep=-2pt,fill=white,minimum size=5mm] (2) at (2,0) {\footnotesize $2$};
\node[circle,draw=black,inner sep=-2pt,fill=white,minimum size=5mm] (3) at (3,0) {\footnotesize $3$};
\node[circle,draw=black,inner sep=-2pt,fill=white,minimum size=5mm] (4) at (4,0) {\footnotesize $4$};
\node[circle,draw=black,inner sep=-2pt,fill=white,minimum size=5mm] (6) at (6,0) {\footnotesize $6$};
\node[circle,draw=black,inner sep=-2pt,fill=white,minimum size=5mm] (8) at (8,0) {\footnotesize $8$};
\node[circle,draw=black,inner sep=-2pt,fill=white,minimum size=5mm] (10) at (10,0) {\footnotesize $10$};
\node[circle,draw=black,inner sep=-2pt,fill=white,minimum size=5mm] (11) at (11,0) {\footnotesize $11$};
\draw (-1) to node[rectangle,draw=none,fill=none,midway,auto] {\scriptsize $\leq\!3$} (0);
\draw (0) to node[rectangle,draw=none,fill=none,midway,auto] {\scriptsize $\leq\!2$} (1);
\draw (1) to node[rectangle,draw=none,fill=none,midway,auto] {\scriptsize $\leq\!2$} (2);
\draw (2) to node[rectangle,draw=none,fill=none,midway,auto] {\scriptsize $\leq\!2$} (3);
\draw (3) to node[rectangle,draw=none,fill=none,midway,auto] {\scriptsize $\leq\!2$} (4);
\draw (4) to node[rectangle,draw=none,fill=none,midway,auto] {\scriptsize $\leq\!2$} (6);
\draw (6) to node[rectangle,draw=none,fill=none,midway,auto] {\scriptsize $\leq\!2$} (8);
\draw (8) to node[rectangle,draw=none,fill=none,midway,auto] {\scriptsize $\leq\!2$} (10);
\draw (10) to node[rectangle,draw=none,fill=none,midway,auto] {\scriptsize $\leq\!3$} (11);
\draw (0) -- (0,1) --node[rectangle,draw=none,fill=none,midway,auto] {\scriptsize $\leq\!1$} (8,1) -> (8);
\draw (0) -- (0,-0.5) --node[rectangle,draw=none,fill=none,pos=0.3,below] {\scriptsize $\leq\!1$} (1.90,-0.5) -> (2);
\draw (2) -- (2,-0.5) --node[rectangle,draw=none,fill=none,midway,below] {\scriptsize $\leq\!1$} (6,-0.5) -> (6);
\draw (1) -- (1,0.5) --node[rectangle,draw=none,fill=none,midway,auto] {\scriptsize $\leq\!1$} (3,0.5) -> (3);
\draw (1) -- (1,-1) --node[rectangle,draw=none,fill=none,midway,below] {\scriptsize $\leq\!1$} (10,-1) -> (10);
\draw (4) -- (4,0.5) --node[rectangle,draw=none,fill=none,midway,auto] {\scriptsize $\leq\!1$} (10,0.5) -> (10);
\end{tikzpicture}}
%\\\vspace{0.5cm}
\subfigure[\label{fig:flow-c}]{
\begin{tikzpicture}[scale=1,->]
\draw[clip,draw=none] (-1.35,-1.3) rectangle (11.3,1.8);
%%%%% backbone
\node[circle,draw=black,inner sep=-2pt,fill=white,minimum size=5mm] (-1) at (-1,0) {\footnotesize $-1$};
\node[circle,draw=black,inner sep=-2pt,fill=white,minimum size=5mm] (0) at (0,0) {\footnotesize $0$};
\node[circle,draw=black,inner sep=-2pt,fill=white,minimum size=5mm] (1) at (1,0) {\footnotesize $1$};
\node[circle,draw=black,inner sep=-2pt,fill=white,minimum size=5mm] (2) at (2,0) {\footnotesize $2$};
\node[circle,draw=black,inner sep=-2pt,fill=white,minimum size=5mm] (3) at (3,0) {\footnotesize $3$};
\node[circle,draw=black,inner sep=-2pt,fill=white,minimum size=5mm] (4) at (4,0) {\footnotesize $4$};
\node[circle,draw=black,inner sep=-2pt,fill=white,minimum size=5mm] (6) at (6,0) {\footnotesize $6$};
\node[circle,draw=black,inner sep=-2pt,fill=white,minimum size=5mm] (8) at (8,0) {\footnotesize $8$};
\node[circle,draw=black,inner sep=-2pt,fill=white,minimum size=5mm] (10) at (10,0) {\footnotesize $10$};
\node[circle,draw=black,inner sep=-2pt,fill=white,minimum size=5mm] (11) at (11,0) {\footnotesize $11$};
\draw[very thick] (-1) to node[rectangle,draw=none,fill=none,midway,auto] {\scriptsize $3$} (0);
\draw[very thick] (0) to node[rectangle,draw=none,fill=none,midway,auto] {\scriptsize $2$} (1);
\draw (1) to node[rectangle,draw=none,fill=none,midway,auto] {\scriptsize $0$} (2);
\draw (2) to node[rectangle,draw=none,fill=none,midway,auto] {\scriptsize $0$} (3);
\draw[very thick] (3) to node[rectangle,draw=none,fill=none,midway,auto] {\scriptsize $1$} (4);
\draw (4) to node[rectangle,draw=none,fill=none,midway,auto] {\scriptsize $0$} (6);
\draw[very thick] (6) to node[rectangle,draw=none,fill=none,midway,auto] {\scriptsize $1$} (8);
\draw[very thick] (8) to node[rectangle,draw=none,fill=none,midway,auto] {\scriptsize $1$} (10);
\draw[very thick] (10) to node[rectangle,draw=none,fill=none,midway,auto] {\scriptsize $3$} (11);
\draw (0) -- (0,1) --node[rectangle,draw=none,fill=none,midway,auto] {\scriptsize $0$} (8,1) -> (8);
\draw[very thick] (0) -- (0,-0.5) --node[rectangle,draw=none,fill=none,pos=0.3,below] {\scriptsize $1$} (1.90,-0.5) -> (2);
\draw[very thick] (2) -- (2,-0.5) --node[rectangle,draw=none,fill=none,midway,below] {\scriptsize $1$} (6,-0.5) -> (6);
\draw[very thick] (1) -- (1,0.5) --node[rectangle,draw=none,fill=none,midway,auto] {\scriptsize $1$} (3,0.5) -> (3);
\draw[very thick] (1) -- (1,-1) --node[rectangle,draw=none,fill=none,midway,below] {\scriptsize $1$} (10,-1) -> (10);
\draw[very thick] (4) -- (4,0.5) --node[rectangle,draw=none,fill=none,midway,auto] {\scriptsize $1$} (10,0.5) -> (10);
\end{tikzpicture}}
\caption{An example for the reduction of Problem~\ref{prob:decision-version} to a max flow problem. In Fig.~\ref{fig:flow-a} an instance $(S,k,t)$ of Problem~\ref{prob:decision-version} consisting of 5 intervals, where we assume $k = 3$ and $t = 1$. One solution is obtained by removing the interval $[0,8)$. In Fig.~\ref{fig:flow-b} the graph $G_{S,k,t}$ whose arcs are labeled by capacities. In Fig.~\ref{fig:flow-b} a max-flow of value $3$ in $G_{S,k,t}$; the arc labels now indicate their flow value. Arc $(0,8)$ has flow value 0.\label{fig:flow-reduction}}
\end{figure}
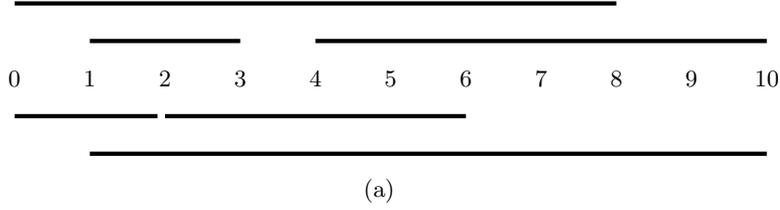
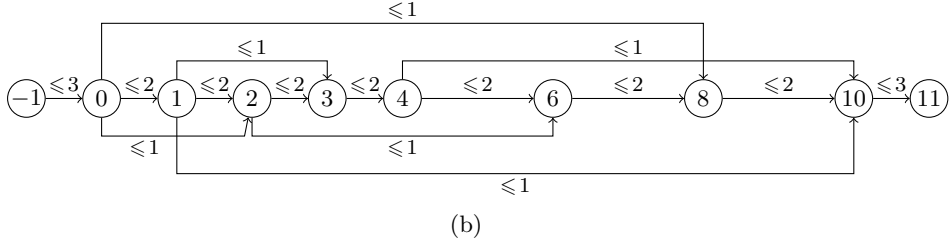
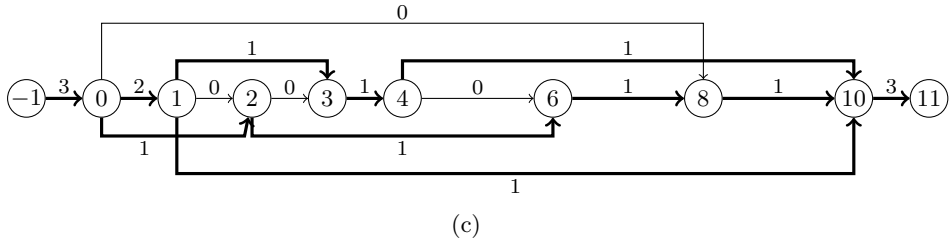

Theorem~\ref{thm:flow-reduction} below shows that computing the max-flow on $G_{S,k,t}$ is equivalent to solving Problem~\ref{prob:decision-version}. The main idea is that the flow passing through an interval arc is equivalent to the selection of that arc in the solution $S'$. The capacity $k-t$ imposed on the backbone arcs not incident to $s-1$ or to $f+1$ implies that for any $p \in [s,f)$ we have at most $k-t$ intervals not covering $p$, thus at least $t$ intervals covering $p$.

\begin{theorem}
Problem~\ref{prob:decision-version} admits a solution on an instance $(S,k,t)$ if and only if the max-flow in $G_{S,k,t}$ has value $k$, and this solution can be retrieved from any integral max-flow on $G_{S,k,t}$. 
\label{thm:flow-reduction}
\end{theorem}

\begin{proof}
We prove the forward implication first. Let $S' \subseteq S$ be a solution to an instance $(S,k,t)$ of Problem~\ref{prob:decision-version}. See also Fig.~\ref{fig:flow-c} for an example. We construct a flow $\f$ in $G_{S,k,t}$ by first assigning $\f(s-1,s) = \f(f,f+1) = k$. Since the capacity of these two arcs is $k$, and since there is no other arc out-going from $s-1$ or in-coming to $f+1$, this will imply that $\f$ is a max-flow in $G_{S,k,t}$. For any interval arc $(s_i,f_i)$, we set $f(s_i,f_i)$ to 1 if and only if the corresponding interval $[s_i,f_i)$ belongs to $S'$. Finally, let $(i,j)$ be any backbone arc different from $(s-1,s)$ and $(f,f+1)$.
Since all $p \in [i,j)$ are covered by the same number of intervals in $S'$, say $t_{i,j}$, and $t_{i,j} \leq t$, we set $f(i,j) = k - t_{i,j}$. 

So far we have obtained that $\f$ satisfies the capacity constraints on $G_{S,k,t}$. It remains to show that the flow conservation property holds for every vertex other than the source and the sink. Let $i$ be such a vertex, and let $(i,j)$ be its out-going backbone arc, and let $(\ell,i)$ be its in-coming backbone arc. The value of the flow out-going from $i$ equals $k - t_{i,j}$ plus the number of intervals in $S'$ with $i$ as left extremity. This equals $k$ minus the number of intervals strictly covering $i$. Similarly, the value of the flow in-coming to $i$ equals $k - t_{\ell,i}$ plus the number of intervals ending at $i$. This also equals $k$ minus the number of intervals strictly covering $i$, thus showing that the flow conservation property holds for $i$.

For the reverse implication, let $\f$ be an integral max-flow in $G_{S,k,t}$ of value $k$ (such an integral flow exists because all capacities are integers). 
% [Do we need to argue the existance of such integral flow? It seems that in the reverse implication we start assuming the existence of such a flow] WE MUST START THE CONSTRUCTION WITH AN INTEGRAL FLOW. IF THE CAPACITIES WERE NOT INTEGERS, SUCH A FLOW DIDN'T EXIST. 
The solution $S' \subseteq S$ to problem Problem~\ref{prob:decision-version} consists of those intervals $[s_i,f_i)$ such that the corresponding interval arc $(s_i,f_i)$ has flow value $1$. Let $(i,j)$ be an arbitrary backbone arc in $G_{S,k,t}$ not incident to the source or to the sink. We need to show that all $p \in [i,j)$ are covered by at least $t$ intervals of $S'$ and by at most $k$ intervals of $S'$. Point $p$ is covered by at most $k$ intervals because $f$ has value $k$. Point $p$ is covered by at least $t$ intervals because the capacity of the backbone arc $(i,j)$ is $k-t$.
\end{proof}

Observe that $G_{S,k,t}$ has $O(n)$ vertices and arcs. Thus, the specialized max-flow algorithm from~\cite{Orlin:2013:MFO:2488608.2488705} applies, leading to the following corollary.

\begin{corollary}
Problem~\ref{prob:decision-version} is solvable in time $O(n^2/\log n)$ by solving the max-flow problem on $G_{S,k,t}$.
\end{corollary}

We can use the above corollary to solve the maximization problem (Problem~\ref{prob:main-problem}) by a doubling/binary search technique as follows. We apply the corollary for $t=1$, $t=2$, $t=4$, $\ldots$, until finding $t^+$ where for $t^-:=2t^+$ the decision problem does no longer have a solution. Binary search on decision problem instances with $t\in [t^+..t^-]$ gives the minimum coverage $t=\mathtt{OPT}$ of the optimal solution to the maximization problem.

\begin{corollary}
Problem~\ref{prob:main-problem} is solvable in time $O(n^2\log \mathtt{OPT}/\log n)$, where $\mathtt{OPT}$, $\mathtt{OPT}\leq k$, is the minimum coverage of an optimal solution.
\label{cor:main-problem-complexity}
\end{corollary}

\section{Tailored max-flow algorithm\label{sec:tailored-flow}}
% !TEX root = paper.tex

Now we give a tailored max-flow algorithm to our problem to achieve a better running time when $k=o(n/\log n)$, based on the Ford-Fulkerson~\cite{FF} max-flow algorithm. Recall that this classical textbook algorithm (see e.g. \cite[Section 26.2]{CLRS}) finds an augmenting path in the residual network and adjusts the flow network along the same path so as to increase the total flow. When there is no augmenting path left in the residual network, the flow found is maximum. Assuming integral capacities (so that the flow increases at least by one unit each augmentation step), the running time is $O(|E||\f^*|)$, where $E$ is the set of arcs of the flow network and $|\f^*|$ is the value of the maximum flow.

Now consider running Ford-Fulkerson on an instance resulting from the reduction of Sec.~\ref{sec:max-flow}. We observe that flow $k-t$ can be sent through the backbone arcs from source to sink. Thus, we can directly initialize the network with a flow of value $k-t$, and start running Ford-Fulkerson from that initial feasible flow. We need at most $t\leq k$ augmentation steps each requiring $O(n)$ time, and thus we obtain the following result.

\begin{theorem}
Problem~\ref{prob:decision-version} is solvable in time $O(nk)$.
\end{theorem}

Using again the above doubling/binary search algorithm on the decision problem we solve the maximization problem:

\begin{corollary}
Problem~\ref{prob:main-problem} is solvable in time $O(n k \log \mathtt{OPT})$, where $\mathtt{OPT}$, $\mathtt{OPT}\leq k$, is the minimum coverage of an optimal solution.
\label{cor:main-problem-tailored-complexity}
\end{corollary}

\section{An $O(n \log n)$ time approximation algorithm\label{sec:2-approximation}}
% !TEX root = paper.tex

Since for large $k$ the best complexity we achieve for our interval scheduling problem is still almost quadratic, we also study a way to find approximate solutions: In this section we present an approximation algorithm running in time $O(n\log n)$, with approximation ratio $\frac{k}{\lfloor k/2 \rfloor}$. For $k$ even, this is a 2-approximation algorithm. First, we extend some concepts introduced in Section~\ref{sec:pruning}, and then make some preliminary observations. Then we describe the algorithm, and finally show how it can be implemented so that it achieves the stated running time.

\begin{sloppypar}
Let us extend the definition of minimum coverage to intervals, so that $\mincov([s_i,f_i)) = min_{p \in [s_i,f_i)}\cov(p)$. 
When an interval has minimum coverage smaller or equal to $\lfloor k/2 \rfloor$ 
we say that such an interval is \emph{crucial}; otherwise, we call it 
\emph{expendable}. The following result is the key idea behind the approximation algorithm.
\end{sloppypar}

\begin{lemma}
  Given an input $(S,k)$ to the interval scheduling maximizing minumum coverage problem, and a point $p$,
  if $\cov_S(p) = k' > k$, there at least $k'-k$ intervals covering $p$
  that are expendable.
\label{lemma:crucial_intervals}
\end{lemma}

\begin{proof}
We proceed by contradiction:
Let us assume that there are $k' > k$ intervals that cover $p$ and that all of them are \emph{crucial}.
For every \emph{crucial} interval, there is at least one point contained in it that has coverage smaller than
or equal to $\lfloor k/2 \rfloor$. Let us call these \emph{supporting points}. 
It is not possible that $p$ is a supporting point, so those must be either larger or smaller 
than $p$. If we separate them among those that are larger than $p$ and those smaller than $p$, one 
of those sets must have at least $\lceil \frac{k+1}{2} \rceil$ intervals. 
Without loss of generality, let us assume that there are at least $\lceil \frac{k+1}{2} \rceil$ intervals
that have a supporting point larger than $p$.
Among those, if we inspect the interval that have the smallest supporting point 
(i.e. the one that is  closest to $p$), this supporting point
must be covered by the totality of the $\lceil \frac{k+1}{2} \rceil$ intervals that have the supporting
point in the same direction. Therefore the coverage of such supporting point must be equal to or
greater than $\lceil \frac{k'}{2} \rceil > \lfloor \frac{k}{2}  \rfloor$, which contradicts the fact that such an interval is
crucial.
\end{proof}

%% Do we need a figure that illustrates the lemma ?

The algorithm proceeds as follows. The original set of intervals $S$ 
will be treated as the set of current candidates.
For every interval, we need to compute its maximum and minimum coverage.
This allows to detect \emph{crucial} intervals, which will never be removed from $S$.
Then, the intervals delimiters are traversed from left to right.
Whenever a delimiter $p$ is found to have coverage $k'$ greater than $k$, then
$k'-k$ intervals covering  $p$ must be removed from the set of candidates.
By Lemma~\ref{lemma:crucial_intervals}, we know that among the $k'$ intervals covering $p$, there are
at least $k'-k$ intervals that are \emph{expendable}, so we can delete those safely.
For every removed interval $[s_i,f_i)$, we need to update the coverage of all the delimiters
that are contained in $[s_i,f_i)$.
%Algorithm \ref{alg:approx} shows the pseudocode of the algorithm.
% Do we want pseudo code?
It is easy to see that the optimal solution to Problem $1$ is bounded by $k$, 
and, because the algorithm never removes \emph{crucial} intervals, 
that the approximation ratio is $\rho = \frac{k}{\lfloor k/2 \rfloor}$.

%%%
Now we show the data structures that allow us to run the algorithm in $O(n \log n)$ time. We build a perfect binary search tree with delimiters of the intervals as leaves. 
Initially, each leaf stores the number of intervals overlapping it, i.e.~the coverage of the delimiter. 
This information can be computed by a sweep from left to right through the delimiters, 
incrementing a counter on the start of an interval and decrementing the counter on the end of an interval, and storing the intermediate counter values to the leaves.
We regard the tree as a one-dimensional range search tree \cite[Section 5.1]{BKOS98}, such 
that internal nodes store keys to allow search towards the leaves by the delimiter. 
We annotate the tree with maximum and minimum of leaf coverages inside each subtree. For leaves, the maximum and minimum correspond to their stored coverage values. 
For internal nodes these values can be computed bottom-up. 
In addition, we annotate each node of the tree with a \emph{balance} counter, initially set to $0$, to support deletion of intervals, as follows.

The approximation algorithm goes through the intervals in the order of their start points. At each such \emph{query} interval $q$, we locate 
a set $V(q)$ of $O(\log n)$ internal nodes that form a partition of the query interval as in \cite[Section 5]{BKOS98}. 
The maximum and minimum coverage encountered at the query interval can be computed by taking 
$\max_{v\in V(q)} v.\mathtt{maxcov}+v.\mathtt{balance}$ and $\min_{v\in V(q)} v.\mathtt{mincov}+v.\mathtt{balance}$, respectively, where
$v.\mathtt{maxcov}$ and $v.\mathtt{mincov}$ are the minimum and maximum coverages of the corresponding subtrees, 
and $v.\mathtt{balance}$, mentioned above, stores a value indicating how much each coverage inside the subtree has changed during earlier steps of the algorithm.
These obtained maximum and minimum coverage values decide if $q$ is deleted or not. If $q$ is deleted, we need to update the coverages in the tree. 
This is done by updating $v.\mathtt{balance}=v.\mathtt{balance}-1$ for all $v \in V(q)$. We propagate the effect of these decrements up to the root, by recomputing maxima and minima on the affected paths, considering 
$v.\mathtt{maxcov}+v.\mathtt{balance}$ and $v.\mathtt{mincov}+v.\mathtt{balance}$ when computing those values. 
Finally, to guarantee that all $v.\mathtt{balance}$ values are maintained correctly, we need to propagate those values down in the tree when querying an interval: 
during the location of a delimiter of a query interval, and moving from parent $p$ of $v$ to $v$, we set $v.\mathtt{balance}=v.\mathtt{balance}+p.\mathtt{balance}$, $w.\mathtt{balance}=w.\mathtt{balance}+p.\mathtt{balance}$, and 
$p.\mathtt{balance}=0$, where $w$ is the other child of $p$. Processing each interval takes $O(\log n)$ time, which proves the running time claim. We have thus obtained the following result.

\begin{theorem}
There is an $O(n \log n)$ time approximation algorithm to Problem~\ref{prob:main-problem} that finds a solution with minimum coverage at least
$\frac{\lfloor k/2 \rfloor}{k}\mathtt{OPT}$, where $\mathtt{OPT}$ is the minimum coverage of an optimal solution.
\label{thm:2approx}
\end{theorem}

\section*{Acknowledgements}

We wish to thank Chao Xu for bringing to our attention the specialized flow algorithms that are suitable for the inputs resulting from our reduction.
 
This work was partially supported by the Academy of Finland grants 284598 to VM, VS, and DV, and 274977 to AT.

\bibliographystyle{plain}
\bibliography{biblio.bib}

\end{document}